\documentclass[conference, letter, 10pt,citestyle=numeric-comp]{IEEEtran}

\makeatletter
\def\endthebibliography{%
  \def\@noitemerr{\@latex@warning{Empty `thebibliography' environment}}%
  \endlist
}
\makeatother

\usepackage{cite}

\usepackage{amsmath,amssymb,amsfonts}
\usepackage{graphicx}
\usepackage{textcomp}
\usepackage{xcolor}

\usepackage{amsthm}
\usepackage[ruled]{algorithm}
\usepackage[noend]{algorithmic}
\usepackage{alltt}
\usepackage{caption}
\usepackage{paralist}
\usepackage{nicefrac}
\usepackage{booktabs}
\usepackage{xspace}
\usepackage{mathtools}
\usepackage{marvosym}
\usepackage{wasysym}
\usepackage{verbatim}
\usepackage{subfigure}
\usepackage{hyperref}
\usepackage{multirow}
\usepackage[capitalize]{cleveref}
\usepackage[per-mode=symbol,detect-all]{siunitx}
\usepackage{graphics}
\usepackage{wrapfig}
\usepackage{enumitem}
\usepackage{stackengine}
\usepackage[normalem]{ulem}
\usepackage{mathpartir}

\newtheorem{definition}{Definition}
\newtheorem{theorem}{Theorem}

\theoremstyle{remark}
\newtheorem*{remark}{Remark}

\usepackage{caption}
\usepackage{tikz,ifthen,pgfplots}
\usetikzlibrary{arrows,trees,backgrounds,automata,shapes,decorations,plotmarks,fit,calc,positioning,shadows,chains}
\tikzstyle{every pin edge}=[<-,shorten <=1pt]
\tikzstyle{neuron}=[circle,fill=black!25,minimum size=17pt,inner sep=0pt]
\tikzstyle{input neuron}=[neuron, fill=green!50]
\tikzstyle{output neuron}=[neuron, fill=red!50]
\tikzstyle{hidden neuron}=[neuron, fill=blue!50]
\tikzstyle{annot} = [text width=4em, text centered]
\usetikzlibrary{calc}

\newcommand{\relu}{\text{ReLU}\xspace}

\newcommand{\vect}[1]{\langle#1\rangle}

\newcommand{\sat}{\texttt{SAT}}
\newcommand{\unsat}{\texttt{UNSAT}}
\newcommand{\timeout}{\texttt{TIMEOUT}}

\newcommand{\ac}{\mathsf{A}}

\DeclareMathOperator{\polarityconstraint}{\mathbf{choosePhase}}
\DeclareMathOperator{\flipPhase}{\mathbf{flipPhase}}
\DeclareMathOperator{\getUnfixedRelus}{\mathbf{getUnfixedReLUs}}

\DeclareMathOperator{\Solve}{\mathbf{solve}}
\DeclareMathOperator{\dequeue}{dequeue}
\DeclareMathOperator{\enqueue}{enqueue}

\DeclareMathOperator{\notempty}{notEmpty}
\DeclareMathOperator{\Partition}{\mathbf{partition}}

\renewcommand{\gg}{\texttt{gg}\xspace}

\newcommand{\dnc}{S\textup{\&}C\xspace}
\newcommand{\ggMarabou}{\textrm{gg-Marabou}\xspace}
\newcommand{\DnCMarabou}{\dnc-Marabou\xspace} 
\newcommand{\infraThread}{\texttt{thread}\xspace}
\newcommand{\infraLocal}{\texttt{gg-local}\xspace}
\newcommand{\infraLambda}{\texttt{gg-lambda}\xspace}

\usepackage{bussproofs}

\newcommand{\xhdr}[1]{{\noindent\bfseries #1}.}

\newif\ifcomments
\commentstrue

\newif\ifoutline
\outlinefalse

\newif\iflong
\longtrue

\newcommand{\Comment}[2]{{\color{#1}{$\curlyvee$}\ifcomments\marginpar{\small\raggedright\color{#1} #2}}\fi}

\newcommand{\aleks}[1]{\Comment{violet}{\textbf{Aleks:} #1}}

\newcommand{\todo}[1]{\Comment{red}{\textbf{[TODO]} #1}}

\renewcommand{\paragraph}[1]{\vspace{1mm}\noindent{\bf #1}\ }

\title{Parallelization Techniques for\\ Verifying Neural Networks}

\DeclareRobustCommand*{\IEEEauthorrefmark}[1]{%
  \raisebox{0pt}[0pt][0pt]{\textsuperscript{\footnotesize #1}}%
}

\author{\IEEEauthorblockN{Haoze Wu\IEEEauthorrefmark{1},
Alex Ozdemir\IEEEauthorrefmark{1},
Aleksandar Zelji\'c\IEEEauthorrefmark{1},
Kyle Julian\IEEEauthorrefmark{1},
Ahmed Irfan\IEEEauthorrefmark{1},
Divya Gopinath\IEEEauthorrefmark{2},\\
Sadjad Fouladi\IEEEauthorrefmark{1},
Guy Katz\IEEEauthorrefmark{5},
Corina Pasareanu\IEEEauthorrefmark{3,4}, and
Clark Barrett\IEEEauthorrefmark{1}}
  \IEEEauthorblockA{\IEEEauthorrefmark{1}Stanford University, USA.
    \IEEEauthorrefmark{2}NASA Ames, KBR Inc.
    \IEEEauthorrefmark{3}NASA Ames, Moffett Field, CA.}
  \IEEEauthorblockA{\IEEEauthorrefmark{4}Carnegie Mellon University, USA.
    \IEEEauthorrefmark{5}The Hebrew University of Jerusalem, Israel.}}

\begin{document}
\maketitle
\begin{abstract}

  Inspired by recent successes of parallel techniques for solving
  Boolean satisfiability, we investigate a set of strategies and heuristics to 
  leverage parallelism and improve the scalability of neural network
  verification. We present a general description of the Split-and-Conquer partitioning algorithm,
  implemented within the Marabou framework, and discuss its parameters and heuristic
  choices. In particular, we explore two novel partitioning strategies, that
  partition the input space or the phases of the neuron activations,
  respectively. We introduce a branching heuristic and a direction heuristic
  that are based on the notion of polarity. We also introduce a highly parallelizable
  pre-processing algorithm for simplifying neural network verification problems.
  An extensive experimental evaluation shows the benefit of these techniques on
  both existing and new benchmarks. A preliminary experiment ultra-scaling
  our algorithm using a large distributed cloud-based platform also shows
  promising results.
\end{abstract}

\section{Introduction}%
\label{sec:intro}

Recent breakthroughs in machine learning, specifically the rise of \emph{deep
  neural networks (DNNs)}~\cite{deeplearning}, have expanded the horizon of
real-world problems that can be tackled effectively. Increasingly, complex
systems are created using machine learning models~\cite{doi:10.2514/1.G003724}
instead of using conventional engineering approaches. Machine learning models
are trained on a set of (labeled) examples, using algorithms that allow the model
to capture their properties and generalize them to unseen inputs. In practice,
DNNs can significantly outperform hand-crafted systems, especially in fields
where precise problem formulation is challenging, such as image
classification~\cite{DBLP:conf/nips/KrizhevskySH12}, speech
recognition~\cite{hinton2012deep} and game playing~\cite{silver2016mastering}.

Despite their overall success, the black-box nature of DNNs calls into question
their trustworthiness and hinders their application in safety-critical domains.
These limitations are exacerbated by the fact that DNNs are known to be vulnerable to
{\em adversarial perturbations}, small modifications to the inputs that lead to wrong
responses from the network~\cite{DBLP:journals/corr/SzegedyZSBEGF13}, and real-world attacks
have already been carried out against safety-critical deployments of
DNNs~\cite{DBLP:journals/corr/CisseANK17,
  DBLP:conf/iclr/KurakinGB17a}.
One promising approach for addressing these concerns is the use of formal
methods to certify and/or obtain rigorous guarantees about DNN behavior.





Early work in DNN formal verification~\cite{DBLP:conf/cav/PulinaT10,
  DBLP:journals/aicom/PulinaT12} focused on translating DNNs and their
properties into formats supported by existing verification tools like general-purpose
\emph{Satisfiability Modulo Theories} (SMT) solvers (e.g.,
Z3~\cite{z3}, CVC4~\cite{cvc4}). However, this approach was limited to small toy
networks (roughly tens of nodes).

More recently, a number of DNN-specific approaches and solvers, including
Reluplex~\cite{KaBaDiJuKo17Reluplex}, ReluVal~\cite{wang2018formal}, Neurify~\cite{WangPWYJ18F}, Planet~\cite{ehlers2017formal},
and Marabou~\cite{katz2019marabou}, have been proposed and developed.
These techniques scale to hundreds or a few thousand nodes.  While a
significant improvement, this is still several orders of magnitude fewer than
the number of nodes present in many real-world applications.  Scalability thus
continues to be a challenge and the subject of active research. 

Inspired by recent successes with parallelizing
SAT solvers~\cite{DBLP:conf/hvc/HeuleKWB11, DBLP:conf/sat/HeuleKM16}, we
propose a set of strategies and heuristics
for leveraging parallelism to improve the scalability of neural
network verification.
The paper makes the following contributions:
\begin{inparaenum}
\item We present a divide-and-conquer algorithm, called Split-and-Conquer (\dnc), for neural
  network verification that is parameterized by different partition strategies and
  constraint solvers (Sec.~\ref{sec:dnc}).
\item We describe two partitioning strategies for this algorithm
  (Sec.~\ref{sec:partition}): one that works by partitioning the input domain
  and a second one that performs case splitting based on the activation
  functions in the neural network. The first strategy was briefly
  mentioned in the Marabou tool paper~\cite{katz2019marabou}; we describe it in
  detail here. The second strategy is new.
\item We introduce the notion of \emph{polarity} and use it to refine the
  partitioning (Sec.~\ref{sec:polarity});
\item We introduce a highly parallelizable pre-processing algorithm that significantly
  simplifies verification problems (Sec.~\ref{sec:lookahead});
\item We show how polarity can additionally be used to speed up satisfiable
  queries (Sec.~\ref{sec:biased}); and
\item We implement the techniques in the Marabou framework and
  evaluate on existing and new neural network
  verification benchmarks from the aviation domain. We also perform an
  \emph{ultra-scalability} experiment using cloud computing
  (Sec.~\ref{sec:experiments}). Our experiments show that the new and improved
  Marabou can outperform the previous version of Marabou as well as other state-of-the-art verification tools such
  as Neurify, especially on perception networks with a large number of inputs.
\end{inparaenum}
%
We begin with preliminaries, review related work in Sec.~\ref{sec:related}, and conclude in Sec.~\ref{sec:concl}.

\section{Preliminaries}
\label{sec:prelim}

In this section, we briefly review neural networks and their formalization, as
well as the Reluplex algorithm for verification of neural networks.

\subsection{Formalizing Neural Networks}

\paragraph{Deep Neural Networks.}
A feed-forward \emph{Deep Neural Network} (DNN) consists of a sequence of
layers, including an input layer, an
output layer, and one or more hidden layers in between.
Each non-input layer comprises multiple \emph{neurons}, whose values can be
computed from the outputs of the preceding layer. Given an assignment of
values to inputs, the output of the DNN can be
computed by iteratively computing the values of neurons in each layer.
Typically, a neuron's value is determined by computing an affine
function of the outputs of the neurons in the previous layer and then
applying a non-linear function, known as an \emph{activation
  function}. A popular activation function is the Rectified Linear Unit (ReLU),
  defined as $ReLU(x)=max(0, x)$ (see~\cite{DBLP:conf/icml/NairH10,
  DBLP:conf/nips/KrizhevskySH12, maas2013rectifier}). In this paper, we focus on
DNNs with ReLU activation functions; thus the output of each
neuron is computed as $ReLU(w_1 \cdot v_1 + \ldots w_n \cdot v_n +b)$, where $v_1
\ldots v_n$ are the values of the previous layer's neurons, $w_1 \ldots w_n$ are
the weight parameters, and $b$ is a bias parameter associated with the neuron. A
neuron is \emph{active} or in the \emph{active phase}, if its output is positive;
otherwise, it is \emph{inactive} or in the \emph{inactive phase}.
%

\begin{figure}[t]
\begin{center}
\includegraphics[scale=0.15]{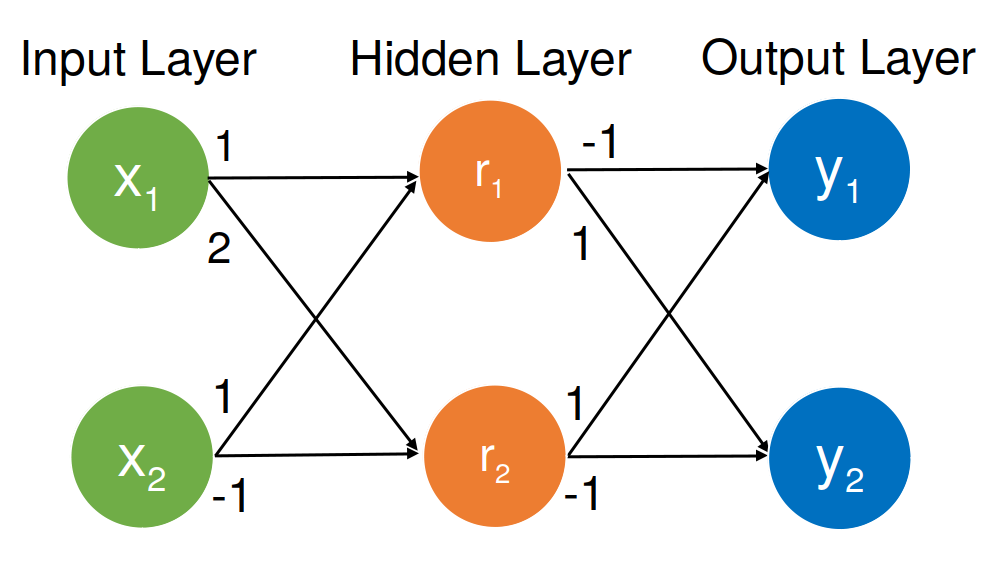}
\end{center}
\vspace{-4mm}
\caption {A small feed-forward DNN $\mathcal{N}$. \label{fig:example}}
\vspace{-4mm}
\end{figure}

\paragraph{Verification of Neural Networks.}
A neural network verification problem has two components: a neural
network $N$, and a property $P$. $P$ is often of the form $P_{in} \Rightarrow
P_{out}$, where $P_{in}$ is a formula over the inputs of $N$ and $P_{out}$ is a
formula over the outputs of $N$.
Typically, $P_{in}$ defines an input region $I$,
and $P$ states that for each point in $I$, $P_{out}$ holds for the
output layer. Given a query like this, a verification tool tries to find a
counter-example: an input point $i$ in $I$, such that when applied to $N$,
$P_{out}$ is false over the resulting outputs.  $P$ holds only if such a
counter-example does not exist. 

The property to be verified may arise from the specific domain where the network is deployed.
For instance, for networks that are used as controllers in 
an unmanned aircraft collision avoidance system (e.g., the ACAS Xu
networks~\cite{KaBaDiJuKo17Reluplex}), we would expect
them to produce sensible advisories according to the location and
the speed of the intruder planes in the vicinity. On the other hand, there are also properties that
are generally desirable for a neural network. One such property is
\emph{local adversarial robustness} \cite{katz2017towards},
which states that a small norm-bounded input perturbation should not cause major spikes 
in the network's output.
%
%
More generally, a property may be an arbitrary formula over input values,
output values, and values of hidden layers---such problems arise for example in the investigation of the neural networks'
explainability \cite{inv}, where one wants to check whether the activation of a
certain \relu $r$ implies a certain output behavior (e.g., the neural network
always predicts a certain class).
%
%
The verification of neural networks with \relu\ functions is decidable and
NP-Complete~\cite{KaBaDiJuKo17Reluplex}. As with many other verification problems,
scalability is a key challenge.

\paragraph{VNN Formulas.}
We introduce the notion of VNN  (Verification of Neural Network) formulas to formalize Neural Network verification
queries. Let $\mathcal{X}$ be a set of variables. A \emph{linear constraint} is
of the form
$\sum_{x_i\in \mathcal{X}}a_i x_i\bowtie b,$
where $a_i, b$ are rational constants, and $\bowtie\ \in \{\leq, \geq,=\}$. A
\emph{\relu~constraint} is of the form
$\relu(x_i) = x_j,$
where $x_i, x_j \in
\mathcal{X}$.
\begin{definition}
A \emph{VNN formula} $\phi$ is a conjunction of linear constraints
and $\relu$ constraints.
\end{definition}
%
A feed-forward neural network can be encoded as a VNN formula as follows.
Each $\relu$ $r$ is represented by introducing a pair of input/output variables $r_b,
r_f$ and then adding a \relu constraint \relu{}($r_b$) = $r_f$.
We refer to $r_b$ as the \emph{backward-facing variable}, and it
is used to connect $r$ to the preceding layer.
$r_f$ is called the \emph{forward-facing variable} and is
used to connect $r$ to the next layer. The weighted
sums are encoded as linear constraints.



In general, a property could be any formula $P$ over the variables used to
represent $\mathcal{N}$.  To check whether $P$ holds on $\mathcal{N}$, we simply conjoin the
representation of $\mathcal{N}$ with the negation of $P$ and use a constraint solver to
check for satisfiability.  $P$ holds iff the constraint is unsatisfiable.

Note
that a solver for VNN formulas can solve a property
$P$ only if the negation
of $P$ is also a VNN formula.  We assume this is the case in this paper, but
more general properties can be handled by decomposing $\neg P$ into a disjunction of
VNN formulas and checking each separately (or, equivalently, using a DPLL($T$)
approach~\cite{NieuwenhuisOT06}).  This works as long as the
atomic constraints are linear.  Non-linear constraints (other than ReLU) are beyond the scope of this paper.

\subsection{The Reluplex Procedure}
The Reluplex procedure~\cite{KaBaDiJuKo17Reluplex} is a sound, complete and
terminating algorithm that decides the satisfiability of a VNN formula. The
procedure extends the Simplex algorithm---a standard efficient decision
procedure for conjunctions of linear constraints---to handle $\relu$
constraints. At a high level, the algorithm iteratively searches for an
assignment that satisfies all the linear constraints, but treats the \relu
constraints lazily in the hope that many of them will be irrelevant for proving
the property.
Once a satisfying assignment for linear constraints is found, the \relu constraints
are evaluated. If all the \relu constraints are satisfied, a model is
found and the procedure concludes that the VNN formula is satisfiable. However,
some \relu constraints may be violated and need to be fixed.
There are two ways to fix a violated \relu constraint $r$:
\begin{inparaenum}
\item \emph{repair the assignment} by updating the assignment of forward-facing $r_f$ or
  backward-facing variable $r_b$ to satisfy $r$, or 
\item \emph{case split} by considering separate cases for each phase of $r$
  (adding the appropriate constraints in each case).
\end{inparaenum}
In both cases, the search continues using the Simplex algorithm, in the first
with a soft correction via assignment update and in the second by adding hard
constraints to the linear problem. Lazy handling of \relu{}s is achieved by the
threshold parameter $t$ --- the number of times a \relu is repaired before the
algorithm performs a case split. In~\cite{KaBaDiJuKo17Reluplex}, this parameter
was set to 20, but even more eager splitting is beneficial in
some cases. The Reluplex algorithm also uses bound propagation to fix \relu{}s to
one phase whenever possible.

In this paper, we explore heuristic choices behind the two options to
handle violated \relu constraints. In the case of assignment repair, the
question is which variable assignment, $r_f$ or $r_b$, to modify (often both are
possible). We refer to the strategy used to make this decision as the
\textit{direction heuristic}, and we discuss direction heuristics, especially in
the context of parallel solving in Sec.~\ref{sec:biased}. For case
splitting, the question is which \relu constraint to choose. We
refer to the strategy used for making this decision as the \textit{branching
  heuristic}. We explore branching heuristics and their application to
parallelizing the algorithm in Sec.~\ref{sec:partition}
and Sec.~\ref{sec:polarity}.

\section{\dnc{}: Parallelizing the Reluplex Procedure}
\label{sec:dnc}

In this section, we present a parallel algorithm called \emph{Split-and-Conquer}
(or simply \dnc{}) for solving VNN formulas, using the Reluplex procedure and
an iterative-deepening strategy. We discuss two partitioning
strategies:  input interval splitting and ReLU case splitting. 

\begin{remark} A divide-and-conquer approach with an
input-splitting strategy was described in the Marabou tool paper
\cite{katz2019marabou}, albeit briefly and informally. We provide here a
more general framework, which includes new techniques and
heuristics, described in detail below.
\end{remark}

\subsection{The \dnc{} algorithm}

\begin{figure}[t]
\begin{center}
\includegraphics[scale=0.18]{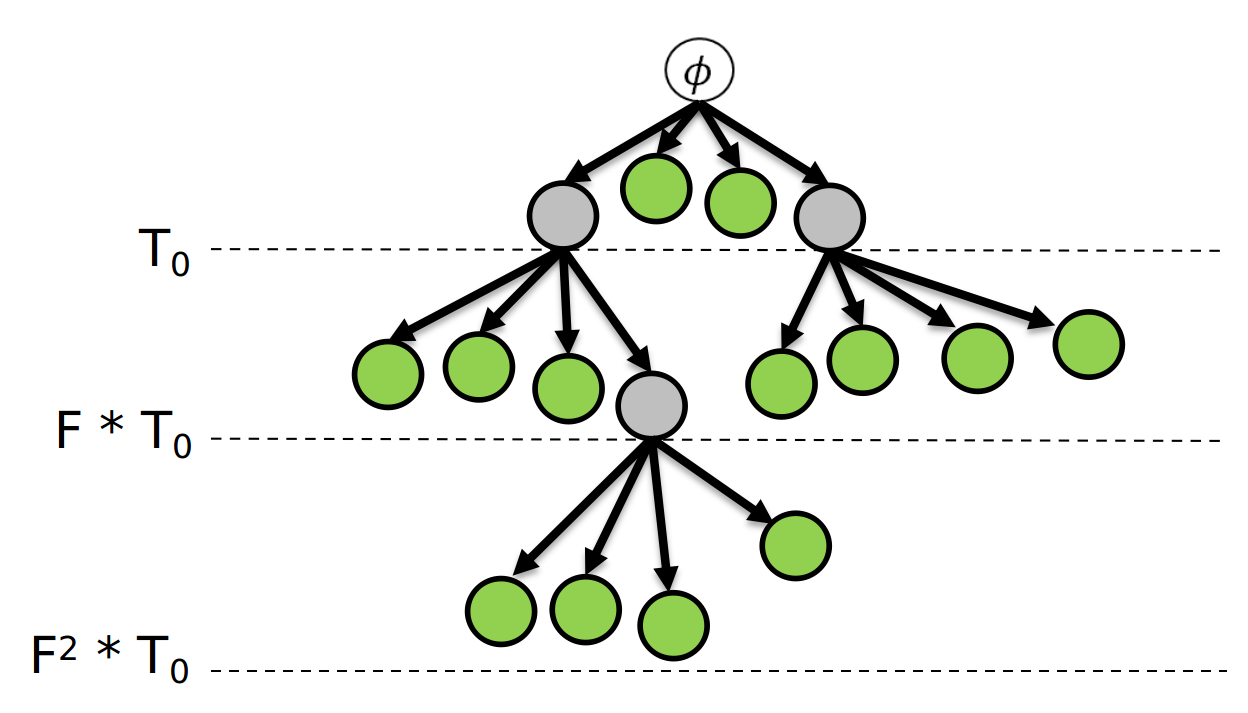}
\end{center}
\caption {An execution of the \dnc{} algorithm.\label{fig:dnc}}
\end{figure}

The \dnc{} algorithm partitions an input problem into several sub-problems (that
are ideally easier to solve) and tries to solve each sub-problem within a given
time budget. If solving a problem exceeds the time budget, that problem is
further partitioned
and the resulting sub-problems are allocated an
increased time budget.
Fig.~\ref{fig:dnc} shows solving of problem $\phi$
as a tree, where the root of the tree denotes the original problem.
Sub-problems that exceed their allotted time budget are partitioned, becoming inner
nodes, and leaves are sub-problems solved within their time budget.
A formula $\phi$ is satisfiable if some leaf is satisfiable. If
the partitioning is \emph{exhaustive}, that is: $\phi \coloneqq \bigvee_{\phi_i
  \in \Partition(\phi,n)}\phi_i,$ for any $n>1$, then $\phi$ is unsatisfiable
iff all the leaves are unsatisfiable.

The pseudo-code of the \dnc{} algorithm is shown in Algorithm~\ref{alg:dnc},
which can be seen as a framework parameterized by the {partitioning
  heuristic} and the underlying {solver}. Details of these parameters are
abstracted away within the $\Partition$ and $\Solve$ functions respectively and
will be discussed in subsequent sections. The \dnc{} algorithm takes as input
the VNN formula $\phi$ and the following parameters: initial number of
partitions $N_0$, initial timeout $T_0$, number of partitions $N$, and the
timeout factor $F$. During solving, \dnc{} maintains a queue $Q$ of
$\langle$query, timeout$\rangle$ pairs, which is initialized with the partition
$N_0 \coloneqq \langle \phi, T_0 \rangle$. While the queue is not empty, the
next pair $\langle\phi', t\rangle$ is retrieved from it, and the query $\phi'$
is solved with time budget $t$. If $\phi'$ is satisfiable, then the original
query $\phi$ is satisfiable, and \sat{} is returned. If $\phi'$ times out,
$\Partition(\phi',N)$ creates $N$ sub-problems of $\phi'$, each of which is enqueued
with an increased time budget $t\cdot F$. If the sub-problem $\phi'$ is
unsatisfiable, no special action needs to be taken. If $Q$ becomes empty, the
original query is unsatisfiable and the algorithm returns \unsat{}.
Note that the main loop of the algorithm naturally lends itself to
parallelization, since the $\Solve$ calls are mutually independent and
query-timeout pairs can be asynchronously enqueued and dequeued.
%
%

\begin{algorithm}[t]
\caption{Split-and-Conquer}
\label{alg:dnc}
\small
\begin{algorithmic}
\STATE {\bfseries Input:} query $\phi$, initial partition size $N_0$, initial
timeout $T_0$, partition size $N$, timeout factor $F$ 
\STATE {\bfseries Output:} \sat{}/\unsat{}
\FOR{$\psi \in {\Partition}(\phi, N_0)$}
\STATE $Q.\enqueue(\langle \psi, T_0 \rangle)$
\ENDFOR
\WHILE{$Q.\notempty()$} 
\STATE $\langle \phi', t \rangle \leftarrow Q.\dequeue()$
\STATE $result \leftarrow \Solve(\phi', t)$
\IF{$result = \sat{}$}
\STATE {\bf return} \sat{}
\ELSIF{$result = \timeout{}$}
\FOR{$\psi \in \Partition(\phi', N)$}
\STATE $Q.\enqueue(\langle \psi, t \cdot F \rangle)$
\ENDFOR
\ENDIF
\ENDWHILE
\STATE {\bf return} \unsat{}
\end{algorithmic}
\end{algorithm}

We state without proof the following result, which is a well-known property of
such algorithms.

\begin{theorem} \label{thm:soundcomplete}
  The Split-and-Conquer$(\phi, N_0, T_0, N, F )$ algorithm is sound and complete
  if the following holds:
\begin{inparaenum}
  \item the $\Solve$ function is sound and complete; and
  \item the $\Partition$ function is exhaustive.
\end{inparaenum} 
\end{theorem}

\noindent
In addition, with modest assumptions on $\Solve$ and $\Partition$, and with $F >
1$, the algorithm can be shown to be terminating. In particular, it
is terminating for the instantiations we consider below.
The \dnc{} algorithm can be tailored to the available computing resources
(e.g., number of processors) by specifying the number of initial splits $N_0$. The
other three search parameters of \dnc{} specify the dynamic behavior of the
algorithm, e.g. if $T_0$ and $F$ are small, or
if $N$ is large, then new sub-queries are created
frequently, which entails a more aggressive \dnc{} strategy (and vice versa).
Notice that we can completely discard the dynamic aspect of \dnc{} by setting
the initial timeout to be $\infty$.

A potential downside of the algorithm is that each call to $\Solve$ that times
out is essentially wasted time, overhead above and beyond the useful work
needed to solve the problem. Fortunately, as the following theorem shows,
the number of wasted calls is bounded.

\begin{theorem}
  When Algorithm~\ref{alg:dnc} runs on an unsatisfiable formula with $N\le N_0$,
  the fraction of calls to $\Solve$ that time out is less than $\frac{1}{N}$.
\end{theorem}%
\begin{proof}
Consider first the case when $N = N_0$.  We can view \dnc{}'s UNSAT proof as constructing an $N$-ary tree, as
  shown in Fig.~\ref{fig:dnc}.
  The $\ell$ leaf nodes are calls to $\Solve$ that do not time out.
  The $t$ non-leaves are calls to $\Solve$ that do time out.
  Since this is a tree,
    the total number of nodes $n$
    is one more than the number of edges.
  Since each query that times out has an edge to each of its $N$ sub-queries,
    the number of edges is $Nt$.
  Thus we have
  \(
    n = Nt + 1
  \)
  which can be rearranged to show the fraction of queries that time out:
  $ \frac{t}{n} = \frac{1-\nicefrac{1}{n}}{N} < \frac{1}{N} $. %
If $N < N_0$, then let $k=N_0-N$.  The number of nodes is then $n = Nt + k +
1$, and the result follows as before.
\end{proof}

\subsection{Partitioning Strategies}
\label{sec:partition}


A partitioning strategy 
specifies how to decompose a VNN formula to
produce (hopefully easier) sub-problems.

A \relu is \emph{fixed} when the bounds on the backward-facing or forward-facing
variable either imply that the \relu is active or imply that the \relu is
inactive. Fixing as many \relu{}s as possible reduces the complexity of the
resulting problem.

With these concepts in mind, we present two strategies:
\begin{inparaenum}
\item \emph{input-based partitioning} creates case splits over the ranges of
  input variables, relying on bound propagation to fix \relu{}s, whereas 
\item \emph{\relu{}-based partitioning} creates case splits that fix
  the phase of \relu{}s directly.
\end{inparaenum}
Both strategies are exhaustive, ensuring soundness and completeness of the \dnc
algorithm (by Theorem~\ref{thm:soundcomplete}). The \emph{branching heuristic}
which determines the choice of input variable, respectively \relu, on which to split,
can have a significant impact
on performance.
%
%
The branching heuristic should keep the total runtime of the
sub-problems low as well as achieve a good \emph{balance} between them. To
illustrate, suppose the sub-problems created by 
splitting $\relu{}_1$ take 10 and 300 seconds to solve, whereas
those created by splitting $\relu{}_2$ take 150 and 160 seconds to solve.
Though the total solving time is the same, the more balanced split, on $\relu{}_2$, 
results in shorter wall-clock time (given two parallel workers).

If most splits led to easier and balanced sub-formulas,
then \dnc{} would perform well, even without a carefully-designed branching
heuristic. However, we have observed that this is not the case for many possible splits:
the time taken to solve one (or both!) of the
sub-problems generated by such splits is comparable to that required by the
original formula (or even worse).  Therefore, an effective branching heuristic
is crucial. We describe two such heuristics below.

\paragraph{Input-based Partitioning.}
This simple partitioning strategy performs case splits over the range of an
input variable.
%
As an example, consider a VNN formula $\phi \coloneqq \phi' \land (-
2 \leq x_1 \leq 1) \land ( -2 \leq x_2 \leq 2 )$, where $x_1$ and $x_2$ are the two
input variables of a neural network encoded by $\phi'$.
Suppose we call $\Partition(\phi,2)$ using the input-splitting strategy.
The choice is between splitting on the range of $x_1$ or the range of $x_2$.
If we choose $x_1$,
the result is two sub-formulas, $\phi_1$ and $\phi_2$,
where:
$ \phi_1 \coloneqq \phi' \land { \bf(-2 \leq x_1< -0.5)} \land (-2
\leq x_2 \leq 2)$ and
$\phi_2 \coloneqq \phi' \land { \bf(-0.5 \leq x_1\leq 1)} \land (-2
\leq x_2 \leq 2)$.
%
An obvious heuristic is to choose the input with 
largest range.  A more complex heuristic was introduced in
\cite{katz2019marabou}.  It samples the network repeatedly, which yields
considerable overhead.  In fact, both of these heuristics perform reasonably well on
benchmarks with only a few inputs (the ACAS Xu benchmarks, for example).
Unfortunately, regardless of the heuristic used,
this strategy suffers from the ``curse of dimensionality'' --- with a
large number of inputs it becomes increasingly difficult to fix \relu{}s by
splitting the range of only one input variable. Thus, the
input-partitioning strategy does not scale well on such networks (e.g., perception networks), which
often have hundreds or thousands of inputs.

\paragraph{$\relu$-based Partitioning. \label{sec:relusplit}}
A complementary strategy is to partition the search space by fixing \relu{}s directly.
%
%
Consider a VNN formula $\phi \coloneqq \phi' \land {(\relu(x)=y)}$. A call to
$\Partition(\phi,2)$ using the \relu{}-based strategy results in two
sub-formulas $\phi_1$ and $\phi_2$, where
$ \phi_1 \coloneqq \phi' \land \bf {(x\leq 0)} \land (y =
0)$
and
$\phi_2 \coloneqq \phi' \land \bf {(x > 0)} \land (x = y)$.
%
Note that here, $\phi_1$ is capturing the inactive and $\phi_2$ the
active phase of the \relu.
Next, we consider a heuristic for choosing a \relu to split on.

\subsection{Polarity-based Branching Heuristics}%
\label{sec:polarity}
We want to estimate the difficulty of sub-problems created by a partitioning
strategy. One key related metric is the number of bounds that can be tightened
as the result of a \relu{}-split.
As a light-weight proxy for this metric, we propose a metric called
\emph{polarity}.
\begin{definition}
Given the \relu{} constraint $\relu{}(x) = y$, and the bounds $a \leq x \leq b$,
where $a < 0$, and $b > 0$, the polarity of the \relu{} is defined as: $p = \frac{a + b}{b
  - a}$.
\end{definition}
\noindent
Polarity ranges from -1 to 1 and measures the symmetry of a \relu{}'s bounds
with respect to zero.  For example, if we split on a \relu{} constraint with
polarity close to 1, the bound on the forward-facing variable in the active case,
$[0,b]$, will be much wider than in the inactive case, $[a,0]$.
Intuitively, forward bound tightening would therefore result in tighter bounds
in the inactive case. This means the inactive case will probably be much easier than
the active case, so the partition is unbalanced and therefore undesirable.
On the other hand, a \relu{} with a polarity close to 0 is more likely to have
balanced sub-problems.
We also observe that \relu{}s in early hidden layers are more likely to 
produce bound tightening by forward bound propagation (as there are more
\relu{}s that depend on them).

We thus propose a heuristic that picks
the \relu{} whose polarity is closest to 0 among the first $k\%$
unfixed \relu{}s, where $k$ is a configurable parameter.
Note that, in order to compute polarities, we need all input variables to be
bounded, which is a reasonable assumption.






\begin{algorithm}[t]
\caption{Iterative Propagation}
\label{alg:preprocess}
\small
\begin{algorithmic}
\STATE {\bfseries Input:} VNN query $\phi$, timeout $t$
\STATE {\bfseries Output:} preprocessed query $\phi'$.
\STATE $progress\leftarrow\top$; $\phi'\leftarrow\phi$
\WHILE{$progress = \top$} 
\STATE{$progress \leftarrow \bot$}
\FOR{$r$ in $\getUnfixedRelus(\phi')$}
\STATE{$\psi \leftarrow \polarityconstraint(r)$}
\STATE {$result = \Solve(\phi' \land \psi, t)$}
\IF{$result = \unsat{}$}
\STATE {$\psi' \leftarrow \flipPhase(\psi)$}
\STATE {$\phi' \leftarrow \phi' \land \psi'$}
\STATE {$progress \leftarrow \top$}
\ENDIF
\ENDFOR
\ENDWHILE
\STATE {\bf return} $\phi'$
\end{algorithmic}
\end{algorithm}

\subsection{Fixing \relu{} Constraints with Iterative Propagation}
\label{sec:lookahead}

As discussed earlier, the performance of \dnc{} depends heavily on
\relu{} splits that result in balanced sub-formulas. However, sometimes a
considerable portion of \relu{}s in a given neural network cannot be split in
this way. To eliminate such \relu{}s we propose a preprocessing technique called
\emph{iterative propagation}, which aims to discover and fix \relu{}s with 
unbalanced partitions.

Concretely, for each \relu{} in the VNN formula, we temporarily fix the \relu
to one of its phases and then attempt to solve the problem with a short
timeout. The goal is to detect unbalanced and
(hopefully) easy unsatisfiable cases.
Pseudocode is presented in Algorithm \ref{alg:preprocess}. The algorithm takes
as input the formula $\phi$ and the timeout $t$, and, if successful, returns an
equivalent formula $\phi'$ which has fewer unfixed \relu{}s than $\phi$. The
outer loop computes the fixed
point, while the inner loop iterates through the
as-of-yet unfixed \relu{}s. For each unfixed \relu, the $\polarityconstraint$
function yields constraints of the easier (i.e. smaller) phase.
If the solver returns $\unsat{}$, then we can safely fix the
\relu{} to its other phase using the $\flipPhase$ function. We ignore the case
where the solver returns \sat{}, since in practice this only occurs for formulas
that are very easy in the first place.

Iterative propagation complements \dnc{}, because the likelihood of finding balanced
partitions is increased by fixing \relu{}s that lead to unbalanced
partitions. Moreover, iterative propagation is highly parallelizable, as each
\relu-fixing attempt can be solved independently. In
Section~\ref{sec:experiments}, we report results using
iterative propagation as a preprocessing step, though it is possible
to integrate the two processes more closely, e.g., by performing iterative
propagation after every $\Partition$ call.



\subsection{Speeding Up Satisfiable Checks with Polarity-Based Direction Heuristics}
\label{sec:biased}
In this section, we discuss how the polarity metric introduced in
Sec.~\ref{sec:polarity} can be used to solve satisfiable instances
quickly.
When splitting on a \relu, the Reluplex algorithm faces the same choice as the
\dnc algorithm. For unsatisfiable cases, the order in which \relu case splits
are done make little difference on average, but for satisfiable instances, it
can be very beneficial if the algorithm is able to hone in on a 
satisfiable sub-problem. We refer to the strategy for picking which \relu phase
to split on first as the \emph{direction heuristic}.


We propose using the polarity metric to guide the
direction heuristic for \dnc{}.
If the polarity of a branching \relu{} is positive,
then we process the active phase first;
if the polarity is negative, we do the reverse.
Intuitively, formulas with wider bounds are more likely
to be satisfiable, and the polarity direction heuristic prefers the phase corresponding to wider bounds
for the \relu{}'s backward-facing variable.

Repairing an assignment when a \relu is violated can also be guided by polarity (recall the description of the Reluplex procedure from Sec.~\ref{sec:prelim}),
as choosing between forward- or backward-facing variables amounts to choosing which
\relu phase to explore first.
Therefore, we use this same direction heuristic to
guide the choice of forward- or backward-facing variables when
repairing the assignment. 
For example, suppose constraint $\relu(x_b) = x_f$ is part of a VNN formula
$\phi$. Suppose the range of $x_b$ is $[-2,1]$, $\ac(x_b) = -1$ and $\ac(x_f) =
1$, where $\ac$ is the current variable assignment computed by the Simplex algorithm.
To repair this violated \relu{} constraint, we can either assign 0 to $x_f$ or assign 1 to $x_b$. In this case, the \relu\ has negative polarity, meaning the negative phase is associated with wider input bounds, so our
heuristic chooses to set $\ac(x_f)=0$.

We will see in our experimental results (Sec.~\ref{sec:experiments}) that
these direction heuristics improve performance on satisfiable instances.
Interestingly, they also enhance performance on unsatisfiable
instances.


\section{Experimental Evaluation}
\label{sec:experiments}
In this section, we discuss our implementation of the proposed techniques and
evaluate its performance on a diverse set of real-world benchmarks -- safety
properties of control systems and robustness properties of perception models.

\subsection{Implementation}
We implemented the techniques discussed above in Marabou
~\cite{katz2019marabou}, which is an open-source neural network verification tool
implementing the Reluplex algorithm. Marabou is available at
\url{https://github.com/NeuralNetworkVerification/Marabou/}\footnote{The version
  of the tool used in the experiments is available
  at
  \url{https://github.com/NeuralNetworkVerification/Marabou/releases/tag/FMCAD20}.}.
The tool also
integrates the symbolic bound tightening techniques introduced in~\cite{wang2018formal}.
We refer to Marabou running the \dnc{}
algorithm as \dnc{}-Marabou. Two partitioning strategies are supported: the
original input-based partitioning strategy
and our new \relu-splitting strategy. All \dnc configurations use the following
parameters: the initial partition size $N_0$ is the number of available
processors; the initial timeout $T_0$ is 10\% of the network size in seconds;
the number of online partitions $N$ is $4$; and the timeout factor $F$ is $1.5$.
The $k$ parameter for the branching heuristic (see Sec.~\ref{sec:polarity}) is
set to 5. The per-\relu{} timeout for iterative propagation is 2 seconds.
When the input dimension is low ($\leq$ 10), symbolic bound tightening
is turned on, and the threshold parameter \textit{t} (see Sec.~\ref{sec:prelim})
is reduced from 20 to 1.
The parameters were chosen using a grid search on a small subset of benchmarks.

\subsection{Benchmarks}
The benchmark set consists of network-property pairs, with networks from three
different application domains: aircraft collision avoidance (ACAS Xu), aircraft
localization (TinyTaxiNet), and digit recognition (MNIST). Properties include
robustness and domain-specific safety properties.

\paragraph{ACAS Xu.}
The ACAS Xu family of VNN benchmarks was introduced
in~\cite{KaBaDiJuKo17Reluplex} and uses prototype neural networks trained to
represent an early version of the ACAS Xu decision logic~\cite{doi:10.2514/1.G003724}. 
The ACAS Xu benchmarks are composed of 45 fully-connected feed-forward neural
networks, each with 6 hidden layers and 50 \relu{} nodes per layer. The networks
issue turning advisories to the controller of an unmanned aircraft to avoid near
midair collisions. The network has 5 inputs (encoding the relation of the
ownship to an intruder) and 5 outputs (denoting advisories: e.g., weak left,
strong right).
Proving that the network does not produce erroneous advisories is paramount for ensuring safe aviation operation. We consider four realistic properties expected of the 45 networks.
These properties, numbered 1--4, are described in~\cite{KaBaDiJuKo17Reluplex}.

\paragraph{TinyTaxiNet.}
The TinyTaxiNet family contains perception networks used in vision-based
\textit{autonomous taxiing}: the task of predicting the position and orientation
of an aircraft on the taxiway, so that a controller can accurately adjust the
position of the aircraft~\cite{julian2020validation}. The input to the network
is a downsampled grey-scale image of the taxiway captured from a camera on the
aircraft. The network produces two outputs: the lateral distance to the runway
centerline, and the heading angle error with respect to the centerline. Proving
that the networks accurately predict the location of the aircraft even when the
camera image suffers from small noise is safety-critical. This property can be
captured as local adversarial robustness. If the
$k$\textsuperscript{th} output of the network is expected to be $b_k$ for inputs
near $\mathbf{a}$, we can check the unsatisfiability of the following VNN formula:
 \begin{align*}
 (y_k \geq b_k + \epsilon)
 \land 
 \bigwedge_{i=1}^{N}(a_i - \delta \leq x_i \leq a_i + \delta),
 \end{align*}
where $\mathbf{x}$ denotes the actual network input,
$N$ the number of network inputs,
and $y_k$ the
actual $k^\text{th}$ output. 
The network is $(\delta,\epsilon)$-locally
robust on $\mathbf{a}$, only if the formula is unsatisfiable.
The training images are compressed to either 2048 or 128 pixels,
with value range [0,1].
We evaluate the local adversarial robustness of two networks. TaxiNet1 has 2048 inputs, 1
convolutional layer, 2 feedforward layers, and 128 \relu{}s.
TaxiNet2 has 128 inputs, 5 convolutional
layers, and a total of 176 \relu{}s. For each network, we generate 100 local
adversarial robustness queries concerning the first output (distance to the
centerline). For each model, we sample 100 uniformly random images from the training data,
and sample $(\delta,\epsilon)$ pairs uniformly from the set
\(\{
  \vect{0.004,3},\allowbreak
  \vect{0.004,9},\allowbreak
  \vect{0.008,3},\allowbreak
  \vect{0.008,9},\allowbreak
  \vect{0.016,9}
\}\).
Setting $\delta=0.004$ allows a 1 pixel-value perturbation in pixel brightness
along each input dimension, and the units of $\epsilon$ are meters. We chose the
values of the perturbation bounds such that the resulting set contains a
mixture of SAT and UNSAT instances with more emphasis on the latter -- UNSAT problems are considered
more interesting in the verification domain.

\paragraph{MNIST.}
In addition to the two neural network families with safety-critical real-world
applications, we evaluate our techniques on three fully-connected feed-forward
neural networks (MNIST1, MNIST2, MNIST3) trained on the MNIST
dataset~\cite{MNISTWebPage} to classify hand-written digits. Each network has
784 inputs (representing a grey-scale image) with value range [0,1], and
10 outputs (each representing a digit). MNIST1 has 10 hidden layers and 10
neurons per layer; MNIST2 has 10 hidden layers and 20 neurons per layer; MNIST3
has 20 hidden layers and 20 neurons per layer.
While shallower and smaller networks may be sufficient for identifying digits and
are also easier to verify, we evaluate on deeper and larger architectures because we want to
\begin{inparaenum}[1)]
\item stress-test our techniques, and
\item evaluate the effect of moving towards larger perception network sizes
  like those used in more challenging applications.
\end{inparaenum}
We consider \textit{targeted robustness} queries,
which asks whether, for an input $\mathbf{x}$ and an incorrect
output $y'$, there exists a point in the $\ell^{\infty}$
$\delta$-ball around $x$ that is classified as $y'$.
We sample 100 such queries for each network,
by choosing random training images and random incorrect labels.
We choose $\delta$ values evenly from
$\{0.004,\allowbreak 0.008,\allowbreak 0.0016,\allowbreak 0.0032\}$.

%
%
%
%

\begin{table*}[tbh!]
\caption{Evaluation of the Techniques on ACAS Xu, TinyTaxiNet, MNIST \label{fig:results}}
\centering
\begin{footnotesize}
\begin{tabular}{@{}l|rr|rr||rr|rr|rr|rr||rr|rr@{}}
  Bench.    & \multicolumn{2}{c|}{\textbf{M}}  &\multicolumn{2}{c||}{\textbf{I}}  & \multicolumn{2}{c|}{\textbf{R}} & \multicolumn{2}{c|}{\textbf{S}} & \multicolumn{2}{c|}{\textbf{S+D}}
  & \multicolumn{2}{c||}{\textbf{S+P}}  & \multicolumn{2}{c|}{\textbf{S+D+P}}  & \multicolumn{2}{c}{\textbf{Neurify}} \\ \hline
   \lbrack\# inst.\rbrack & \#S & Time & \#S & Time & \#S & Time & \#S & Time& \#S & Time & \#S & Time & \#S & Time& \#S & Time \\
\hline\hline
    {\bf ACAS}&   40 & 17224 & \textbf{45} & \textbf{4884} & \textbf{45} & 5009 & \textbf{45} & \textbf{4884} & \textbf{45} & 5480 & \textbf{45} & 8419 & \textbf{45} & 7244 & 39 & 4167 \\ 
    \lbrack180\rbrack &   101 & 57398 & 130 & 48954 & 125 & 45036 & 130 & 48954 & 131 & 51413 & 130 & 50828 & 131 & 53717 & \textbf{133} & 1438 \\ \hline\hline
    {\bf TinyTaxi.}& 34 & 4591 & 34 & 1815 & 34 & 433& 34 & 433 & 34 & 419 & 34 & 533 & \textbf{35} & 1172 & \textbf{35} & \textbf{88} \\ 
    \lbrack200\rbrack &  141 & 33909 & 110 & 24088 & 147 & 23079 & 147 & 23079 & 147 & 22345 & \textbf{149} & \textbf{20583} & \textbf{149} & 21949 & 146 & 7158 \\ \hline\hline
    {\bf MNIST}&   11 & 2349 & 19 & 13032 & 22 & 9680 & 22 & 9680 & 26 & 11727 & 20 & 9956 & \textbf{29} & 19351 & 27 & 151 \\ 
    \lbrack300\rbrack &   140 & 64418 & 78 & 27134 & 181 & 52776 & 181 & 52776 & 183 & 59195 & 184 & 67625 & \textbf{185} & 68307 & 153 & 10640 \\ \hline\hline
    {\bf All} &   85 & 24164 & 98 & 19731 & 101 & 15122 & 101 & 14997 & 105 & 17626 & 99 & 18908 & \textbf{109} & 27767 & 101 & 4406 \\ 
    \lbrack680\rbrack &   382 & 155725 & 318 & 100176 & 453 & 120891 & 458 & 124809 & 461 & 132953 & 463 & 139036 & \textbf{465} & 143973 & 432 & 19236 \\
    \hline
\end{tabular}

\vspace{2mm}
Number of solved instances (\#S) and run-time in seconds of different
configurations. For each benchmark set, top and bottom rows show data for \\satisfiable
(SAT) and unsatisfiable (UNSAT) instances respectively. The results for configuration \textbf{S} are computed virtually from \textbf{R} and \textbf{I}.
\end{footnotesize}
\end{table*}

\subsection{Experimental Evaluation}
We present the results of the following experiments:
1) Evaluation of each technique's effect on run-time performance of Marabou on
  the three benchmark sets. We also compare against Neurify, a
  state-of-the-art solver on the same benchmarks.
2) An analysis of trade-offs when running iterative propagation pre-processing.
3) Exploration of \dnc\ scalability at a large scale, using cloud
  computing.

\subsubsection{{Evaluation of the techniques on ACAS Xu, TinyTaxiNet, MNIST}
  \label{experiment:all}}
We denote the \relu{}-based partitioning strategy as \textbf{R}, polarity-based
direction heuristics as \textbf{D}, and iterative propagation as \textbf{P}.
We denote as \textbf{S} a hybrid strategy that uses input-based partitioning on
ACAS Xu networks, and \relu{}-based partitioning on perception networks. We run
four combinations of our techniques:
\begin{inparaenum}
\item \textbf{R};
\item \textbf{S+D};
\item \textbf{S+P};
\item \textbf{S+D+P},
\end{inparaenum}
and compare them with two baseline configurations:
\begin{inparaenum}
\item the sequential mode of Marabou (denoted as \textbf{M});
\item \dnc{}-Marabou with its default input-based partitioning strategy (denoted
  as \textbf{I}).
\end{inparaenum}

We compare with Neurify~\cite{WangPWYJ18F}, a state-of-the-art solver, on the
same benchmarks. Neurify derives over-approximations of the output bounds
using techniques such as symbolic interval analysis and linear relaxation. On
ACAS Xu benchmarks, it operates by iteratively partitioning the input region to
reduce error in the over-approximated bounds (to prove UNSAT) and by randomly sampling
points in the input region (to prove SAT). On other networks, Neurify uses
off-the-shelf solvers to handle ReLU-nodes whose bounds are potentially
overestimated. Neurify also leverages parallelism, as different input regions or
linear programs can be checked in parallel.

We run all Marabou configurations and Neurify on a cluster equipped with Intel
Xeon E5-2699 v4 CPUs running CentOS 7.7. 8 cores and 64GB RAM are allocated for
each job, except for the \textbf{M} configuration, which uses 1 processor and
8GB RAM per job. Each job is given a 1-hour wall-clock timeout.

\paragraph{Results.}
Table~\ref{fig:results} shows a breakdown of the number of solved
instances and the run-time for all Marabou configurations and for Neurify. We group
the results by SAT and UNSAT instances. For each row, we highlight the entries
corresponding to the configuration that solves the most instances (ties broken
by run-time). Here are some key observations:
\begin{itemize}[wide, labelwidth=!, labelindent=0pt]
\item[--] On ACAS Xu benchmarks, both input-based partitioning (\textbf{I}) and \relu{}-based
  partitioning (\textbf{R}) yield performance gain compared with the sequential solver (\textbf{M}),
  with \textbf{I} being more effective. On perception networks, \textbf{I} solves significantly fewer
  instances than \textbf{M} while \textbf{R} continues to be effective.

\item[--]
  Comparing the performance of \textbf{S}, \textbf{S+D}, and \textbf{S+P} suggests that 
  the polarity-based direction heuristics and iterative propagation each improve the overall
  performance of \dnc{}-Marabou. Interestingly, the polarity-based heuristic improves
  the performance on not only SAT but also UNSAT instances, suggesting that by
  affecting how \relu{} constraints are repaired, direction heuristics also favorably
  impact the order of \relu{}-splitting.
  On the other hand, iterative propagation alone only improves performance on
  UNSAT instances. \textbf{S+D+P} solves the most instances among all the
  Marabou configurations, indicating that the direction heuristics and iterative
  propagation are complementary to each other.
\item[--] \textbf{S+D+P} solves significantly more instances than Neurify
  overall. While Neurify's strategy on Acas Xu benchmarks allows it to dedicate
  more time on proving UNSAT by rapidly partitioning the input region (thus
  yielding much shorter run-times than \textbf{S+D+P} on that benchmark set), its
  performance on SAT instances is subject to (un)lucky guesses. When it comes to
  perception neural networks that are deeper and have higher input dimensions,
  symbolic bound propagation, on which Neurify heavily relies, becomes more
  expensive and less effective. In contrast, Marabou does not rely solely on
  symbolic interval analysis, but in addition uses interval bound-tightening
  techniques (see \cite{katz2019marabou} for details).
\end{itemize}

\begin{figure}[ht]
  \centering
    \captionsetup{singlelinecheck = false, format= hang, justification=raggedright, labelsep=space}
    \includegraphics[width=8.2cm]{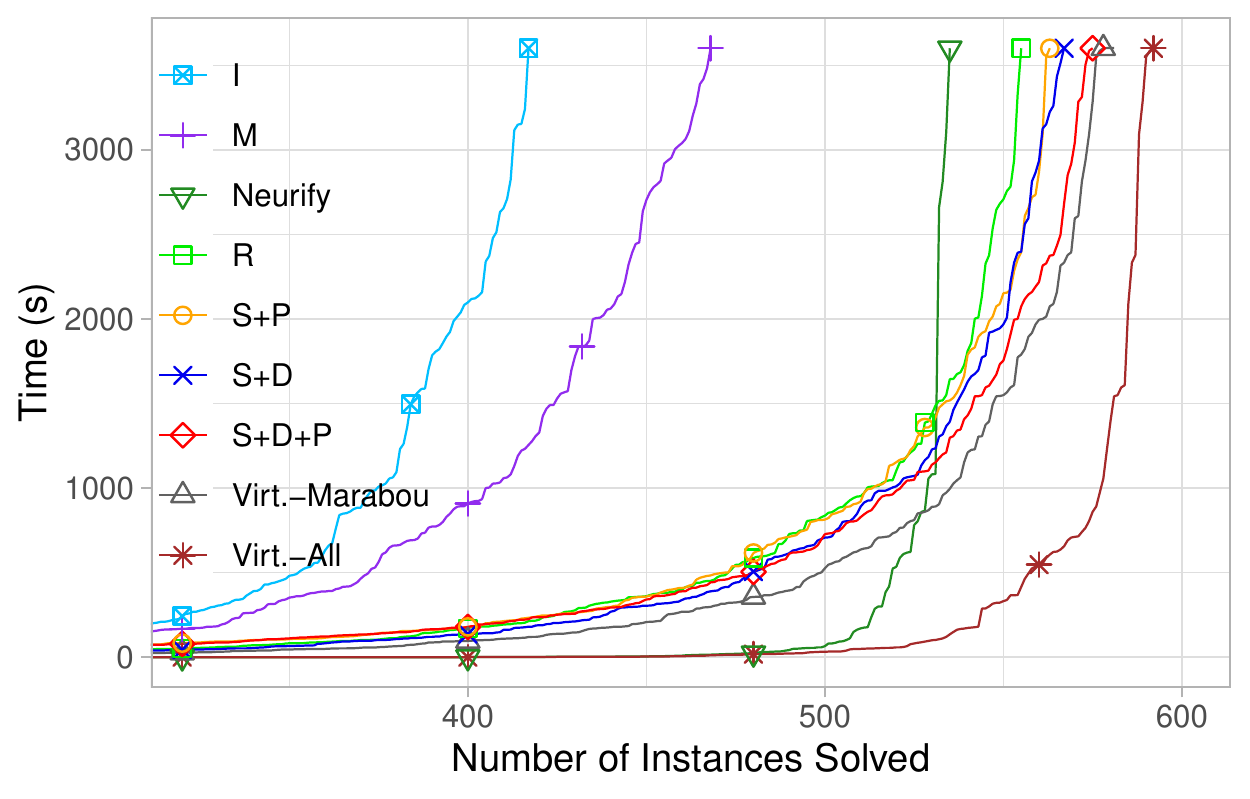}
    \caption{Cactus plot: all solvers + two virtual best configurations.}
    \vspace{4mm}
    \label{fig:cactus}
    \centering
    \includegraphics[width=8cm]{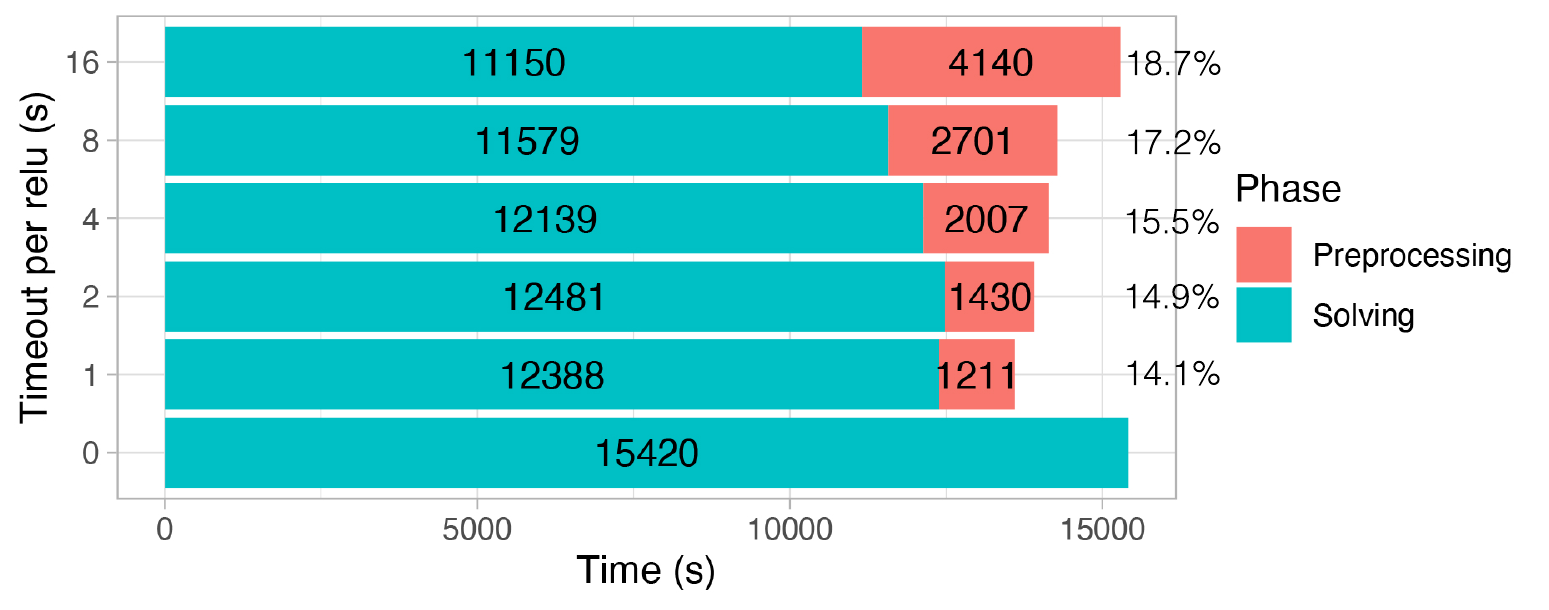}
    \caption{The effect of varying per-\relu\ timeout in preprocessing.\label{fig:prep}}
\end{figure}


Fig.~\ref{fig:cactus} shows a cactus plot of the 6 Marabou configurations and Neurify on all benchmarks.
In this plot, we also include two virtual portfolio configurations: \textbf{Virt.-Marabou} takes
the best run-time among all Marabou configurations for each benchmark, and
\textbf{Virt.-All} includes Neurify in the portfolio. Interestingly,
\textbf{S+D+P}
is outperformed by \textbf{S+D} in the beginning but surpasses \textbf{S+D}
after 500  seconds.
This suggests that iterative propagation creates overhead for easy instances, but benefits the
search in the long run. We also observe that Neurify can solve a subset of the benchmarks very rapidly, but
solves very few benchmarks after 1500 seconds. One possible explanation 
is that Neurify splits the input region and makes solver calls eagerly. While this allows it to resolve
some queries quickly, it also results in rapid (exponential) growth of the number of sub-regions and solver
calls. By contrast, Marabou splits lazily. While it creates overhead sometimes, it results in more solved
instances overall. The \textbf{Virt.-All} configuration solves significantly more instances than
\textbf{Virt.-Marabou}, suggesting that the two procedures are complementary to each other.
We note that the bound tightening techniques presented in Neurify can be potentially
integrated into Marabou, and the polarity-based heuristics and iterative propagation could
also be used to improve Neurify and other VNN tools.

\subsubsection{Costs of Iterative Propagation}
As mentioned in Sec.~\ref{alg:preprocess}, intuitively, the longer the time
budget during iterative propagation, the more \relu{}s should get fixed. To
investigate this trade-off between the number of fixed \relu{}s and the
overhead, we choose a smaller set of benchmarks (40 ACAS Xu benchmarks, 40 TinyTaxiNet
benchmarks, and 40 MNIST benchmarks), and vary the timeout parameter $t$ of
iterative propagation. Each job is run with 32 cores, and a wall-clock timeout
of 1 hour, on the same cluster as in Experiment \ref{experiment:all}.
%
%
%

\paragraph{Results.}
Fig.~\ref{fig:prep} shows the preprocessing time +
solving time of different configurations on commonly solved instances. The percentage next
to each bar represents the average percentage of \relu{}s fixed by iterative propagation.
Though the run-time and unfixed \relu{}s continue to decrease as we invest more in
iterative propagation, performing iterative propagation
no longer provides performance gain when the per-\relu{}-timeout exceeds 8 seconds.

\subsubsection{Ultra-Scalability of \dnc{}}
\DnCMarabou\ runs on a single machine, which intrinsically limits its
scalability to the number of hardware threads. To investigate how the \dnc\
algorithm scales with much higher degrees of parallelism, we implemented it on top of the \gg\
platform~\cite{gg}.

%

The \gg\ platform facilitates expressing
parallelizable computations and executing them.
To use it, the programmer expresses their computation
as a dependency graph of tasks,
where each task is an executable program
that reads and writes files.
The output files can encode the
result of the task,
or an extension to the task graph that must be executed in order to
produce that result.
The \gg\ platform includes tools for executing tasks in parallel.
Tasks can be executed \emph{locally}, using different processes,
or \emph{remotely}, using cloud services such as AWS Lambda \cite{awslambda}.
Since these cloud services offer a high
degree of concurrency with little setup or administration,
\gg\ is a convenient tool for executing massively parallel computations~\cite{gg}.
%
%

In our implementation of the \dnc\ algorithm on top of \gg,
each task runs the base solver with a timeout. If the solver completes, the task
returns the result; otherwise it returns a task graph extension encoding the
division of the problem into sub-queries.
We call this implementation of the \dnc\ algorithm, \textit{\ggMarabou}.

We measure the performance of \dnc\
and \ggMarabou\ at varying levels of
parallelism to establish that they perform similarly and to evaluate the
scalability of the \dnc\ algorithm.
Our experiments use three underlying infrastructures:
\DnCMarabou\ (denoted \infraThread),
\ggMarabou\ executed locally (\infraLocal),
and \ggMarabou\ executed remotely on AWS Lambda~\cite{awslambda} (\infraLambda).
We vary the parallelism level, $p$, from 4 to 16 for the local infrastructures
%
and from 4 to 1000 for \infraLambda.
For \infraLambda, we run 3 tests per benchmark, taking the median time to
mitigate variation from the network.
From the UNSAT ACAS Xu benchmarks which \dnc{}-Marabou can solve in under two
hours using 4 cores, we chose 5 of the hardest instances.
We set $T_0 = \SI{5}{s}$, $F = 1.5$, $N = 2^{\lfloor (5 + \log_2 p) / 3 \rfloor}$
and use the input-based partitioning strategy.

\paragraph{Results.}
Fig.~\ref{fig:scaling:interval} shows how mean runtime (across benchmarks)
varies with parallelism level and infrastructure.
\begin{figure}[t]
  \centering
  \includegraphics[width=2.3in]{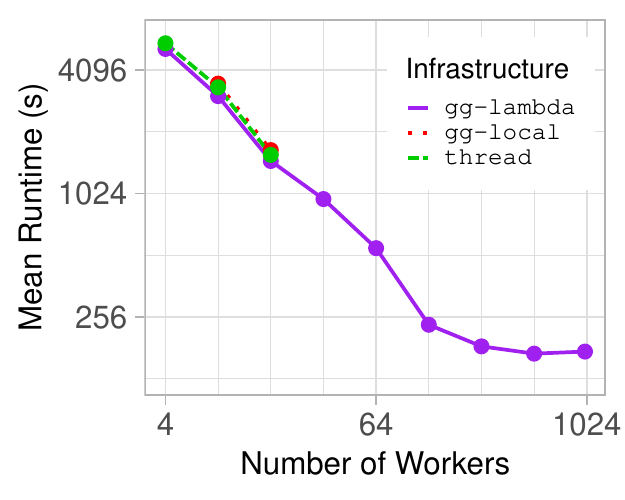}
  \caption{Ultra-Scalability of \dnc.} 
  \label{fig:scaling:interval}
\end{figure}
Our first conclusion from Fig.~\ref{fig:scaling:interval} is that \gg\ does
\textbf{not} introduce significant overhead;
at equal parallelism levels, all infrastructures perform similarly.
Our second conclusion is that \ggMarabou\ scales well up
to over a hundred workers.
This is shown by the constant slope of the runtime/parallelism level line
up to over a hundred workers.
We note that the slope only flattens when total runtime is small:
a few minutes.

\section{Related Work}
\label{sec:related}


Over the past few years, a number of tools for verifying neural network have
emerged and broadly fall into two categories --- precise and abstraction-based
methods. Precise approaches are complete and usually encode the problem as an
SAT/SMT/MILP 
constraint\cite{ehlers2017formal,BDNN,KaBaDiJuKo17Reluplex,katz2019marabou,venus}.
Abstraction-based methods are not necessarily complete and abstract the search
space using intervals~\cite{wang2018formal,WangPWYJ18F} or more complex abstract
domains~\cite{gehr2018ai2,deepz,singh2019abstract}. However, most of these
approaches are sequential, and for details, we refer the reader to the survey by Liu et
al.~\cite{liu2019algorithms}.
%
%
To the best of our knowledge, only Marabou~\cite{katz2019marabou} and
Neurify~\cite{WangPWYJ18F} (and its predecessor ReluVal~\cite{wang2018formal})
leverage parallel computing to speed up verification.
As mentioned in Sec.~\ref{sec:experiments}, Neurify combines symbolic interval analysis
with linear relaxation to compute tighter output bounds and uses off-the-shelf solvers
to derive more precise bounds for \relu{}s. These interval analysis
techniques lend themselves well to parallelization, as independent linear programs
can be created and checked in parallel. By contrast, \dnc-Marabou creates
partitions of the original query and solves them in parallel.
Neurify supports a selection of hard-coded
benchmarks and properties and often requires modifications to
support new properties, while Marabou provides verification support for a wide range
of properties.

Split-and-Conquer is inspired by the \emph{Cube-and-Conquer}
algorithm~\cite{DBLP:conf/hvc/HeuleKWB11}, which targets very hard SAT problems.
Cube-and-Conquer is a divide-and-conquer technique that partitions a Boolean
satisfiability problem into sub-problems by conjoining cubes ---a cube is a
conjunction of propositional literals--- to the original problem and then
employing a conflict-driven SAT solver~\cite{DBLP:series/faia/SilvaLM09} to
solve each sub-problem in parallel. The propositional literals used in cubes are
chosen using look-ahead~\cite{DBLP:series/faia/HeuleM09} techniques.
Divide-and-conquer techniques have also been used to parallelize SMT solving~\cite{parallelSMT,marescotti2016clause}.
Our approach uses similar ideas to those in previous work, but is optimized for the VNN domain.

Iterative propagation is, in part, inspired by the look-ahead techniques. While the latter is
used to partition the search space, the former is used to reduce the overall complexity of the problem.

 	

\section{Conclusions and Future Work}
\label{sec:concl}

In this paper, we presented a set of techniques that leverage parallel
computing to improve the scalability of neural network verification.  We
described an algorithm based on partitioning the verification problem in an
iterative manner and explored two strategies that work by partitioning the
input space or by splitting on \relu{}s,
respectively. We introduced a branching heuristic and a direction heuristic,
both based on the notion of polarity. We also introduced a highly
parallelizable pre-processing algorithm for simplifying neural
network verification problems.
Our experimental evaluation shows the benefit of these techniques on existing and
new benchmarks.  A preliminary experiment with
ultra-scaling using the gg platform on Amazon Lambda also shows promising results.

Future work includes:
\begin{inparaenum}[i)]
\item Investigating more dynamic strategies for choosing hyper-parameters of the \dnc{} framework.
\item Investigating different ways to interleave iterative propagation with \dnc{}.
\item Investigating the scalability of ReLU-based partitioning to high levels of
parallelism.
\item Improving the performance of the underlying solver, Marabou, by integrating
conflict analysis (as in CDCL SAT solvers and SMT solvers) and more
advanced bound propagation techniques such as those used by Neurify.
\item Extending the techniques to handle other piecewise-linear activation
functions such as hard tanh and leaky ReLU, to which the notion of
polarity applies.
\end{inparaenum}


\section*{Acknowledgements}
The project was partially supported by grants from the Binational
Science Foundation (2017662), the Defense Advanced Research Projects Agency
(FA8750-18-C-0099), Ford Motor Company,
the Israel Science Foundation (683/18), and the National Science
Foundation (1814369).


\bibliography{bibli}

\begin{thebibliography}{10}
\providecommand{\url}[1]{#1}
\csname url@samestyle\endcsname
\providecommand{\newblock}{\relax}
\providecommand{\bibinfo}[2]{#2}
\providecommand{\BIBentrySTDinterwordspacing}{\spaceskip=0pt\relax}
\providecommand{\BIBentryALTinterwordstretchfactor}{4}
\providecommand{\BIBentryALTinterwordspacing}{\spaceskip=\fontdimen2\font plus
\BIBentryALTinterwordstretchfactor\fontdimen3\font minus
  \fontdimen4\font\relax}
\providecommand{\BIBforeignlanguage}[2]{{%
\expandafter\ifx\csname l@#1\endcsname\relax
\typeout{** WARNING: IEEEtran.bst: No hyphenation pattern has been}%
\typeout{** loaded for the language `#1'. Using the pattern for}%
\typeout{** the default language instead.}%
\else
\language=\csname l@#1\endcsname
\fi
#2}}
\providecommand{\BIBdecl}{\relax}
\BIBdecl

\bibitem{deeplearning}
I.~Goodfellow, Y.~Bengio, and A.~Courville, \emph{Deep Learning}.\hskip 1em
  plus 0.5em minus 0.4em\relax MIT Press, 2016,
  \url{http://www.deeplearningbook.org}.

\bibitem{doi:10.2514/1.G003724}
\BIBentryALTinterwordspacing
K.~D. Julian, M.~J. Kochenderfer, and M.~P. Owen, ``Deep neural network
  compression for aircraft collision avoidance systems,'' \emph{Journal of
  Guidance, Control, and Dynamics}, vol.~42, no.~3, pp. 598--608, 2019.
  [Online]. Available: \url{https://doi.org/10.2514/1.G003724}
\BIBentrySTDinterwordspacing

\bibitem{DBLP:conf/nips/KrizhevskySH12}
A.~Krizhevsky, I.~Sutskever, and G.~E. Hinton, ``Imagenet classification with
  deep convolutional neural networks,'' in \emph{{NIPS}}, 2012, pp. 1106--1114.

\bibitem{hinton2012deep}
G.~Hinton, L.~Deng, D.~Yu, G.~E. Dahl, A.-r. Mohamed, N.~Jaitly, A.~Senior,
  V.~Vanhoucke, P.~Nguyen, T.~N. Sainath \emph{et~al.}, ``Deep neural networks
  for acoustic modeling in speech recognition: The shared views of four
  research groups,'' \emph{IEEE Signal processing magazine}, vol.~29, no.~6,
  pp. 82--97, 2012.

\bibitem{silver2016mastering}
D.~Silver, A.~Huang, C.~J. Maddison, A.~Guez, L.~Sifre, G.~Van Den~Driessche,
  J.~Schrittwieser, I.~Antonoglou, V.~Panneershelvam, M.~Lanctot \emph{et~al.},
  ``Mastering the game of go with deep neural networks and tree search,''
  \emph{nature}, vol. 529, no. 7587, p. 484, 2016.

\bibitem{DBLP:journals/corr/SzegedyZSBEGF13}
C.~Szegedy, W.~Zaremba, I.~Sutskever, J.~Bruna, D.~Erhan, I.~J. Goodfellow, and
  R.~Fergus, ``Intriguing properties of neural networks,'' in \emph{{ICLR}
  (Poster)}, 2014.

\bibitem{DBLP:journals/corr/CisseANK17}
M.~Ciss{\'{e}}, Y.~Adi, N.~Neverova, and J.~Keshet, ``Houdini: Fooling deep
  structured prediction models,'' \emph{CoRR}, vol. abs/1707.05373, 2017.

\bibitem{DBLP:conf/iclr/KurakinGB17a}
A.~Kurakin, I.~J. Goodfellow, and S.~Bengio, ``Adversarial examples in the
  physical world,'' in \emph{{ICLR} (Workshop)}.\hskip 1em plus 0.5em minus
  0.4em\relax OpenReview.net, 2017.

\bibitem{DBLP:conf/cav/PulinaT10}
L.~Pulina and A.~Tacchella, ``An abstraction-refinement approach to
  verification of artificial neural networks,'' in \emph{{CAV}}, ser. Lecture
  Notes in Computer Science, vol. 6174.\hskip 1em plus 0.5em minus 0.4em\relax
  Springer, 2010, pp. 243--257.

\bibitem{DBLP:journals/aicom/PulinaT12}
------, ``Challenging {SMT} solvers to verify neural networks,'' \emph{{AI}
  Commun.}, vol.~25, no.~2, pp. 117--135, 2012.

\bibitem{z3}
L.~M. de~Moura and N.~Bj{\o}rner, ``{Z3:} an efficient {SMT} solver,'' in
  \emph{{TACAS}}, ser. Lecture Notes in Computer Science, vol. 4963.\hskip 1em
  plus 0.5em minus 0.4em\relax Springer, 2008, pp. 337--340.

\bibitem{cvc4}
C.~W. Barrett, C.~L. Conway, M.~Deters, L.~Hadarean, D.~Jovanovic, T.~King,
  A.~Reynolds, and C.~Tinelli, ``{CVC4},'' in \emph{{CAV}}, ser. Lecture Notes
  in Computer Science, vol. 6806.\hskip 1em plus 0.5em minus 0.4em\relax
  Springer, 2011, pp. 171--177.

\bibitem{KaBaDiJuKo17Reluplex}
G.~Katz, C.~Barrett, D.~Dill, K.~Julian, and M.~Kochenderfer, ``{Reluplex: An
  Efficient SMT Solver for Verifying Deep Neural Networks},'' in \emph{Proc.
  29th Int. Conf. on Computer Aided Verification (CAV)}, 2017, pp. 97--117.

\bibitem{wang2018formal}
S.~Wang, K.~Pei, J.~Whitehouse, J.~Yang, and S.~Jana, ``Formal security
  analysis of neural networks using symbolic intervals,'' in \emph{27th
  $\{$USENIX$\}$ Security Symposium ($\{$USENIX$\}$ Security 18)}, 2018, pp.
  1599--1614.

\bibitem{WangPWYJ18F}
\BIBentryALTinterwordspacing
------, ``Efficient formal safety analysis of neural networks,'' in
  \emph{Advances in Neural Information Processing Systems 31: Annual Conference
  on Neural Information Processing Systems 2018, NeurIPS 2018, 3-8 December
  2018, Montr{\'{e}}al, Canada}, 2018, pp. 6369--6379. [Online]. Available:
  \url{http://papers.nips.cc/paper/7873-efficient-formal-safety-analysis-of-neural-networks}
\BIBentrySTDinterwordspacing

\bibitem{ehlers2017formal}
R.~Ehlers, ``Formal verification of piece-wise linear feed-forward neural
  networks,'' in \emph{International Symposium on Automated Technology for
  Verification and Analysis}.\hskip 1em plus 0.5em minus 0.4em\relax Springer,
  2017, pp. 269--286.

\bibitem{katz2019marabou}
G.~Katz, D.~A. Huang, D.~Ibeling, K.~Julian, C.~Lazarus, R.~Lim, P.~Shah,
  S.~Thakoor, H.~Wu, A.~Zelji{\'c} \emph{et~al.}, ``The marabou framework for
  verification and analysis of deep neural networks,'' in \emph{International
  Conference on Computer Aided Verification}, 2019, pp. 443--452.

\bibitem{DBLP:conf/hvc/HeuleKWB11}
M.~Heule, O.~Kullmann, S.~Wieringa, and A.~Biere, ``Cube and conquer: Guiding
  {CDCL} {SAT} solvers by lookaheads,'' in \emph{Haifa Verification
  Conference}, ser. Lecture Notes in Computer Science, vol. 7261.\hskip 1em
  plus 0.5em minus 0.4em\relax Springer, 2011, pp. 50--65.

\bibitem{DBLP:conf/sat/HeuleKM16}
M.~J.~H. Heule, O.~Kullmann, and V.~W. Marek, ``Solving and verifying the
  boolean pythagorean triples problem via cube-and-conquer,'' in \emph{{SAT}},
  ser. Lecture Notes in Computer Science, vol. 9710.\hskip 1em plus 0.5em minus
  0.4em\relax Springer, 2016, pp. 228--245.

\bibitem{DBLP:conf/icml/NairH10}
V.~Nair and G.~E. Hinton, ``Rectified linear units improve restricted boltzmann
  machines,'' in \emph{{ICML}}.\hskip 1em plus 0.5em minus 0.4em\relax
  Omnipress, 2010, pp. 807--814.

\bibitem{maas2013rectifier}
A.~L. Maas, A.~Y. Hannun, and A.~Y. Ng, ``Rectifier nonlinearities improve
  neural network acoustic models,'' in \emph{Proc. icml}, vol.~30, no.~1, 2013,
  p.~3.

\bibitem{katz2017towards}
G.~Katz, C.~W. Barrett, D.~L. Dill, K.~Julian, and M.~J. Kochenderfer,
  ``Towards proving the adversarial robustness of deep neural networks,'' in
  \emph{FVAV@iFM}, ser. {EPTCS}, vol. 257, 2017, pp. 19--26.

\bibitem{inv}
D.~{Gopinath}, H.~{Converse}, C.~{Pasareanu}, and A.~{Taly}, ``Property
  inference for deep neural networks,'' in \emph{2019 34th IEEE/ACM
  International Conference on Automated Software Engineering (ASE)}, Nov 2019,
  pp. 797--809.

\bibitem{NieuwenhuisOT06}
R.~Nieuwenhuis, A.~Oliveras, and C.~Tinelli, ``Solving sat and sat modulo
  theories: From an abstract davis-- putnam--logemann--loveland procedure to
  dpll({\it t}),'' \emph{J. ACM}, vol.~53, no.~6, pp. 937--977, 2006.

\bibitem{julian2020validation}
K.~D. Julian, R.~Lee, and M.~J. Kochenderfer, ``Validation of image-based
  neural network controllers through adaptive stress testing,'' \emph{arXiv
  preprint arXiv:2003.02381}, 2020.

\bibitem{MNISTWebPage}
``The {MNIST} database of handwritten digits {Home Page},''
  \url{http://yann.lecun.com/exdb/mnist/}.

\bibitem{gg}
\BIBentryALTinterwordspacing
S.~Fouladi, F.~Romero, D.~Iter, Q.~Li, S.~Chatterjee, C.~Kozyrakis, M.~Zaharia,
  and K.~Winstein, ``From laptop to lambda: Outsourcing everyday jobs to
  thousands of transient functional containers,'' in \emph{2019 {USENIX} Annual
  Technical Conference, {USENIX} {ATC} 2019, Renton, WA, USA, July 10-12,
  2019}, 2019, pp. 475--488. [Online]. Available:
  \url{https://www.usenix.org/conference/atc19/presentation/fouladi}
\BIBentrySTDinterwordspacing

\bibitem{awslambda}
``{AWS} lambda,'' \url{https://docs.aws.amazon.com/lambda/index.html}.

\bibitem{BDNN}
N.~Narodytska, S.~Kasiviswanathan, L.~Ryzhyk, M.~Sagiv, and T.~Walsh,
  ``Verifying properties of binarized deep neural networks,'' in
  \emph{Thirty-Second AAAI Conference on Artificial Intelligence}, 2018.

\bibitem{venus}
E.~Botoeva, P.~Kouvaros, J.~Kronqvist, A.~Lomuscio, and R.~Misener, ``Efficient
  verification of relu-based neural networks via dependency analysis.'' in
  \emph{AAAI}, 2020, pp. 3291--3299.

\bibitem{gehr2018ai2}
T.~Gehr, M.~Mirman, D.~Drachsler-Cohen, P.~Tsankov, S.~Chaudhuri, and
  M.~Vechev, ``Ai2: Safety and robustness certification of neural networks with
  abstract interpretation,'' in \emph{2018 IEEE Symposium on Security and
  Privacy (SP)}.\hskip 1em plus 0.5em minus 0.4em\relax IEEE, 2018, pp. 3--18.

\bibitem{deepz}
G.~Singh, T.~Gehr, M.~Mirman, M.~P{\"u}schel, and M.~Vechev, ``Fast and
  effective robustness certification,'' in \emph{Advances in Neural Information
  Processing Systems}, 2018, pp. 10\,802--10\,813.

\bibitem{singh2019abstract}
G.~Singh, T.~Gehr, M.~P{\"u}schel, and M.~Vechev, ``An abstract domain for
  certifying neural networks,'' \emph{Proceedings of the ACM on Programming
  Languages}, vol.~3, no. POPL, pp. 1--30, 2019.

\bibitem{liu2019algorithms}
C.~Liu, T.~Arnon, C.~Lazarus, C.~Barrett, and M.~J. Kochenderfer, ``Algorithms
  for verifying deep neural networks,'' 2019.

\bibitem{DBLP:series/faia/SilvaLM09}
J.~P.~M. Silva, I.~Lynce, and S.~Malik, ``Conflict-driven clause learning {SAT}
  solvers,'' in \emph{Handbook of Satisfiability}, ser. Frontiers in Artificial
  Intelligence and Applications.\hskip 1em plus 0.5em minus 0.4em\relax {IOS}
  Press, 2009, vol. 185, pp. 131--153.

\bibitem{DBLP:series/faia/HeuleM09}
M.~Heule and H.~van Maaren, ``Look-ahead based {SAT} solvers,'' in
  \emph{Handbook of Satisfiability}, ser. Frontiers in Artificial Intelligence
  and Applications.\hskip 1em plus 0.5em minus 0.4em\relax {IOS} Press, 2009,
  vol. 185, pp. 155--184.

\bibitem{parallelSMT}
A.~E. Hyv{\"a}rinen, M.~Marescotti, and N.~Sharygina, ``Search-space
  partitioning for parallelizing smt solvers,'' in \emph{International
  Conference on Theory and Applications of Satisfiability Testing}.\hskip 1em
  plus 0.5em minus 0.4em\relax Springer, 2015, pp. 369--386.

\bibitem{marescotti2016clause}
M.~Marescotti, A.~E. Hyv{\"a}rinen, and N.~Sharygina, ``Clause sharing and
  partitioning for cloud-based smt solving,'' in \emph{International Symposium
  on Automated Technology for Verification and Analysis}.\hskip 1em plus 0.5em
  minus 0.4em\relax Springer, 2016, pp. 428--443.

\end{thebibliography}
\bibliographystyle{IEEEtran}

\newpage
\onecolumn
\section*{\huge{Appendix}}

\subsection{More Details on \ggMarabou}
The \gg\ platform is a tool for expressing
parallelizable computations and executing them. To use it, the programmer expresses
their computation as \emph{task graph}: a dependency graph of tasks, where each
task is an executable program (e.g., a binary or shell script) that reads some
input files and produces some output files. These output files can encode the
result of the task, or an extension to the task graph that must be executed in order to
produce that result. In our implementation of the \dnc\ algorithm on top of \gg,
each task runs the base solver with a timeout. If the solver completes, the task
returns the result, otherwise it returns a task graph extension encoding the
division of the problem into sub-queries.

The local part of the \gg\ experiment is run on a machine with 24 Xeon E5-2687W v4 
CPUs, 132GB RAM, running Ubuntu 20.04.

\subsection{More Details on Evaluation of Techniques}
We present here a more detailed report of the runtime performance of different configurations and
Neurify, as shown in Table \ref{fig:results-full}. We break down the ACAS Xu benchmark
family by properties, and the other two benchmark sets by networks.

Fig.~\ref{fig:scatter} shows the log-scaled pairwise comparisons between different configurations.

\begin{table*}[bh!]
\centering
\begin{footnotesize}
\begin{tabular}{@{}l|rr|rr||rr|rr|rr||rr|rr@{}}
  Bench.    & \multicolumn{2}{c|}{\textbf{M}}  &\multicolumn{2}{c||}{\textbf{I}}  & \multicolumn{2}{c|}{\textbf{R}} & \multicolumn{2}{c|}{\textbf{S+D}}
  & \multicolumn{2}{c||}{\textbf{S+P}}  & \multicolumn{2}{c|}{\textbf{S+D+P}}  & \multicolumn{2}{c}{\textbf{Neurify}} \\ \hline
   \lbrack\# inst.\rbrack & S & T & S & T & S & T & S & T & S & T & S & T & S & T \\
\hline\hline
 ACAS1 &   0 & 0 & 0 & 0 & 0 & 0 & 0 & 0 & 0 & 0 & 0 & 0 & 0 & 0\\ 
    45 &   17 & 32455 & 42 & 37125 & 37 & 33141 & 43 & 40936 & 42 & 36783 & 43 & 40107 & \textbf{45} & \textbf{558} \\ \hline
 ACAS2 &   34 & 17210 & \textbf{39} & \textbf{4863} & \textbf{39} & 4985 & \textbf{39} & 5456 & \textbf{39} & 8228 & \textbf{39} & 7074 & 33 & 4167 \\ 
  45 &   0 & 0 & 4 & 5461 & 4 & 5121 & 4 & 4042 & 4 & 4070 & 4 & 4156 & \textbf{4} & \textbf{88} \\ \hline
  ACAS3 &   \textbf{3} & 9 & \textbf{3} & 10 & \textbf{3} & 12 & \textbf{3} & 13 & \textbf{3} & 91 & \textbf{3} & 83 & \textbf{3} & \textbf{0} \\ 
  45 &   \textbf{42} & 18254 & \textbf{42} & 4900 & \textbf{42} & 4569 & \textbf{42} & 4826 & \textbf{42} & 6571 & \textbf{42} & 6295 & \textbf{42} & \textbf{742} \\ \hline
  ACAS4 &   \textbf{3} & 5 & \textbf{3} & 11 & \textbf{3} & 12 & \textbf{3} & 11 & \textbf{3} & 100 & \textbf{3} & 87 & \textbf{3} & \textbf{0} \\ 
  45 &   \textbf{42} & 6689 & \textbf{42} & 1468 & \textbf{42} & 2205 & \textbf{42} & 1609 & \textbf{42} & 3404 & \textbf{42} & 3159 & \textbf{42} & \textbf{49} \\ \hline 
 {\bf ACAS}&   40 & 17224 & \textbf{45} & \textbf{4884} & \textbf{45} & 5009 & \textbf{45} & 5480 & \textbf{45} & 8419 & \textbf{45} & 7244 & 39 & 4167 \\ 
 180&   101 & 57398 & 130 & 48954 & 125 & 45036 & 131 & 51413 & 130 & 50828 & 131 & 53717 & \textbf{133} & \textbf{1438} \\ \hline\hline
  TinyTaxiNet1 &   \textbf{11} & 1168 & \textbf{11} & 1579 & \textbf{11} & 370 & \textbf{11} & 356 & \textbf{11} & 337 & \textbf{11} & 357 & \textbf{11} & \textbf{85} \\ 
  100 &   84 & 27773 & 63 & 17883 & 89 & 14052 & 89 & 12521 & 89 & 14651 & \textbf{89} & \textbf{14683} & 81 & 7148 \\ \hline
  TinyTaxiNet2 &   23 & 3423 & 23 & 236 & 23 & 63 & 23 & 63 & 23 & 196 & \textbf{24} & 815 & \textbf{24} & \textbf{3} \\ 
  100 &   57 & 6136 & 47 & 6205 & 58 & 9027 & 58 & 9824 & 60 & 5932 & 60 & 7266 & \textbf{65} & \textbf{10} \\ \hline 
  {\bf TinyTaxiNet}& 34 & 4591 & 34 & 1815 & 34 & 433 & 34 & 419 & 34 & 533 & \textbf{35} & 1172 & \textbf{35} & \textbf{88} \\ 
  200 &  141 & 33909 & 110 & 24088 & 147 & 23079 & 147 & 22345 & 149 & 20583 & \textbf{149} & \textbf{21949} & 146 & 7158 \\ \hline\hline
 MNIST1 &   9 & 2178 & 11 & 3658 & 12 & 2190 & 12 & 1412 & 12 & 3682 & \textbf{13} & \textbf{5715} & 8 & 108 \\ 
  100 &   73 & 11880 & 47 & 12387 & \textbf{80} & \textbf{12999} & 80 & 15213 & 80 & 14090 & 80 & 13571 & 54 & 2285 \\ \hline
  MNIST2 &   2 & 171 & 5 & 3494 & 6 & 3246 & 7 & 3787 & 4 & 1782 & 9 & 9140 & \textbf{13} & \textbf{6} \\ 
  100 &   37 & 22069 & 17 & 5698 & 46 & 14576 & 46 & 14833 & \textbf{48} & \textbf{19026} & 47 & 15141 & 45 & 3247 \\ \hline
  MNIST3 &   0 & 0 & 3 & 5880 & 4 & 4244 & \textbf{7} & 6528 & 4 & 4492 & \textbf{7} & \textbf{4496} & 6 & 36 \\ 
  100 &   30 & 30469 & 14 & 9049 & 55 & 25201 & 57 & 29149 & 56 & 34509 & \textbf{58} & \textbf{39595} & 54 & 5108 \\ \hline
 {\bf MNIST}&   11 & 2349 & 19 & 13032 & 22 & 9680 & 26 & 11727 & 20 & 9956 & \textbf{29} & \textbf{19351} & 27 & 151 \\ 
  300 &   140 & 64418 & 78 & 27134 & 181 & 52776 & 183 & 59195 & 184 & 67625 & \textbf{185} & \textbf{68307} & 153 & 10640 \\ \hline\hline
{\bf All} &   85 & 24164 & 98 & 19731 & 101 & 15122 & 105 & 17626 & 99 & 18908 & \textbf{109} & \textbf{27767} & 101 & 4406 \\ 
  680 &   382 & 155725 & 318 & 100176 & 453 & 120891 & 461 & 132953 & 463 & 139036 & \textbf{465} & \textbf{143973} & 432 & 19236 \\ \hline
\end{tabular}
\vspace{2mm}
\caption{Number of solved instances (S) and run time in seconds (T) of different
  configurations. For each family, top and bottom rows show data for satisfiable
  (SAT) and unsatisfiable (UNSAT) instances respectively. \label{fig:results-full}}
\end{footnotesize}
\vspace{-2mm}
\end{table*}

\begin{figure}[ht]
  \centering
    \includegraphics[width=18cm]{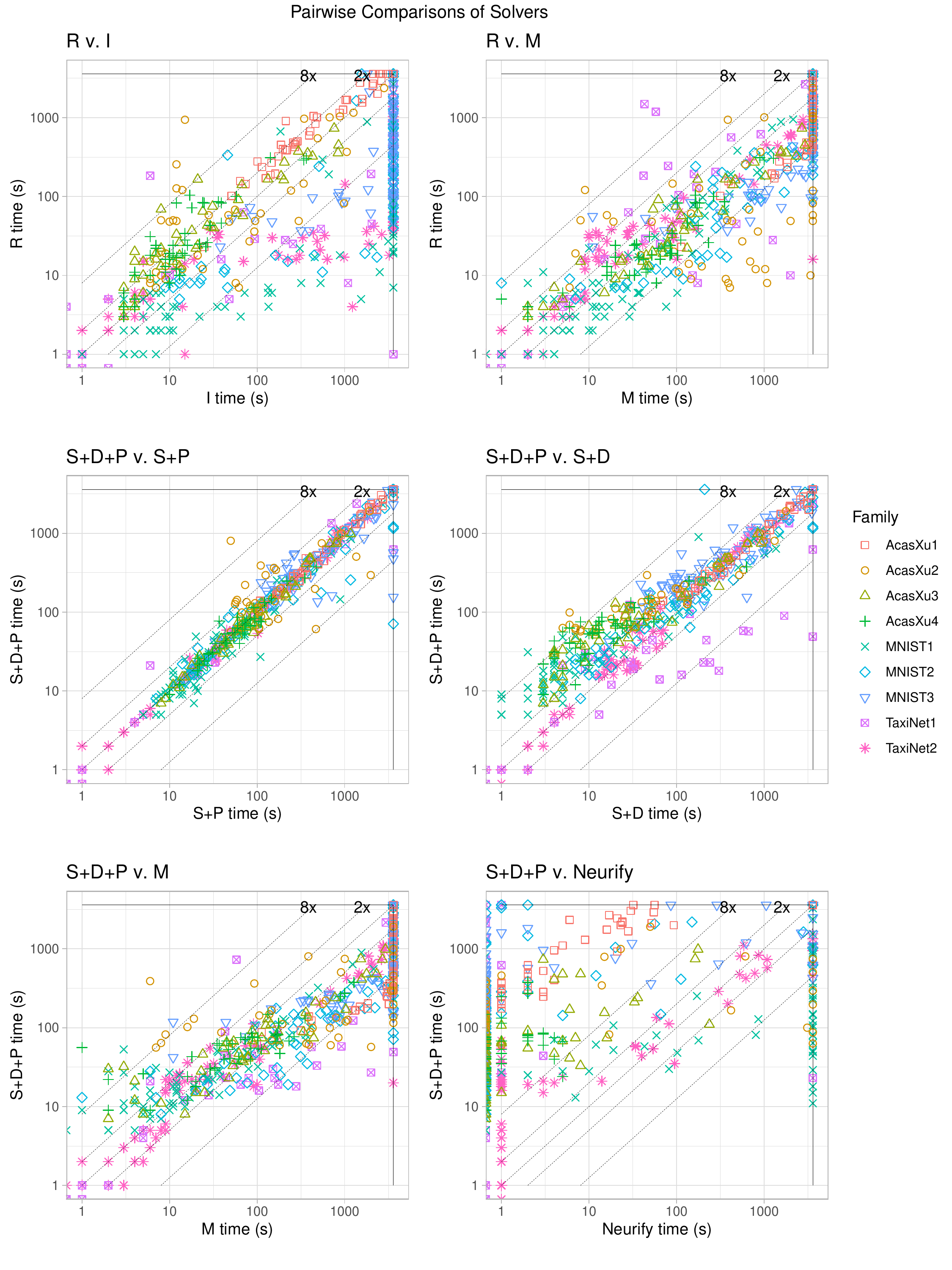}
    \caption{Pairwise comparison between different configurations on all benchmarks.\label{fig:scatter}}
\end{figure}

\end{document}